\documentclass[conference]{IEEEtran}
\IEEEoverridecommandlockouts
\usepackage{cite}
\usepackage{amsmath,amssymb,amsfonts}
\usepackage{graphicx}
\usepackage{textcomp}
\usepackage{xcolor}

\usepackage{colortbl}
\usepackage[normalem]{ulem}

\usepackage{algpseudocode}

\DeclareMathOperator*{\argmin}{arg\,min}
\usepackage[normalem]{ulem}
\useunder{\uline}{\ul}{}
\usepackage{graphicx}
\usepackage{subcaption}
\usepackage{mwe}
\usepackage{comment}
\usepackage{algpseudocode}
\usepackage{lipsum}
\usepackage[export]{adjustbox}
\usepackage{amsthm}

\newtheorem{proposition}{Proposition}

\def\minwrt[#1]{\underset{#1}{\mathrm{minimize }}}
\def\maxwrt[#1]{\underset{#1}{\mathrm{maximize }}}
\def\argminwrt[#1]{\underset{#1}{\mathrm{arg min }}}

\newcommand{\otdist}{S}

\newcommand{\bh}{\mathbf{h}}
\newcommand{\bu}{\mathbf{u}}
\newcommand{\bv}{\mathbf{v}}
\newcommand{\bw}{\mathbf{w}}
\newcommand{\bC}{\mathbf{C}}
\newcommand{\bK}{\mathbf{K}}
\newcommand{\bM}{\mathbf{M}}
\newcommand{\blambda}{\boldsymbol{\lambda}}
\newcommand{\bmu}{\boldsymbol{\mu}}
\newcommand{\bxi}{\boldsymbol{\xi}}
\newcommand{\onevec}{\mathbf{1}}
\newcommand{\bnu}{\boldsymbol{\nu}}

\newcommand{\bhgeo}{\mathbf{h}_{\mathrm{0}}}

\newcommand{\RR}{\mathbb{R}}

\def\BibTeX{{\rm B\kern-.05em{\sc i\kern-.025em b}\kern-.08em
    T\kern-.1667em\lower.7ex\hbox{E}\kern-.125emX}}
\begin{document}
\title{Room Impulse Response Estimation using Optimal Transport: Simulation-Informed Inference}

\author{\IEEEauthorblockN{David Sundström\IEEEauthorrefmark{2}, Anton Björkman\IEEEauthorrefmark{1},
Andreas Jakobsson\IEEEauthorrefmark{2}, and
Filip Elvander\IEEEauthorrefmark{1}}
\IEEEauthorblockA{\IEEEauthorrefmark{1}Dept. of Information and Communications Engineering, Aalto University, Finland\\
}
\IEEEauthorblockA{\IEEEauthorrefmark{2}Dept. of Mathematical Sciences, Lund University, Sweden
}
}

\maketitle

\begin{abstract}
The ability to accurately estimate room impulse responses (RIRs) is integral to many applications of spatial audio processing.
Regrettably, estimating the RIR using ambient signals, such as speech or music, remains a challenging problem due to, e.g., low signal-to-noise ratios, finite sample lengths, and poor spectral excitation.
Commonly, in order to improve the conditioning of the estimation problem, priors are placed on the amplitudes of the RIR.
Although serving as a regularizer, this type of prior is generally not useful when only approximate knowledge of the delay structure is available, which, for example, is the case when the prior is a simulated RIR from an approximation of the room geometry. 
In this work, we target the delay structure itself, constructing a prior based on the concept of optimal transport.
As illustrated using both simulated and measured data, the resulting method is able to beneficially incorporate information even from simple simulation models, displaying considerable robustness to perturbations in the assumed room dimensions and its temperature.
\end{abstract}


\section{Introduction}\vspace{-1mm}
\label{sect:introduction}
Accurate and robust estimation of the room impulse response (RIR) is necessary for many forms of emerging spatial audio applications, including sound zones \cite{Lee2018}, spatial active noise control \cite{Koyama2021}, and rendering for virtual reality \cite{Kenzie2022}. Although the estimation problem is well studied for controlled settings, it remains a challenging problem to estimate an RIR using ambient signals, such as music or speech \cite{Naylor2011}.

Typically, the problem is aggravated by the presence of any movement in the observed sound source (see, e.g., \cite{Sundstrom2024ICASSP}), as well as the inherent characteristics of the often non-stationary source signal itself. In particular, short signal observation, poor spectral excitation, and low signal-to-noise ratio (SNR), makes the problem ill-conditioned. To counter this, multiple approaches to regularize the RIR estimation have been proposed.

In \cite{Waterschoot2008,Ljung2020}, the use of Tikhonov regularization for solving the inverse problem was presented, corresponding to the maximum likelihood estimator when using Gaussian priors on the amplitudes of the RIR. To exploit that the early part of an RIR is sparse under the idealistic assumption of specular reflections, \cite{Lin2006, Crocco2015, Benichoux2014} consider estimation of an RIR using variations of the Lasso regularization. 
Regrettably, reflections within a room has inherit frequency-dependent absorption and diffusion characteristics such that measured RIRs are typically not sparse. As a further alternative, low rank modelling of RIRs have also recently been proposed in \cite{Jalmby2021,Jalmby2023,Jalmby2023_b}. 

\begin{figure}
\centering
\begin{subfigure}{\linewidth}
    \centering
    \includegraphics[width = 7cm]{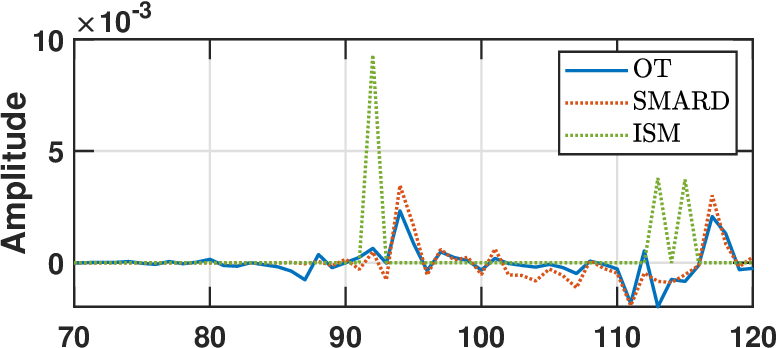}
    \caption{}\label{fig:illustration_ot}
\end{subfigure}

\medskip

\begin{subfigure}{\linewidth}
    \centering
    \includegraphics[width=7cm]{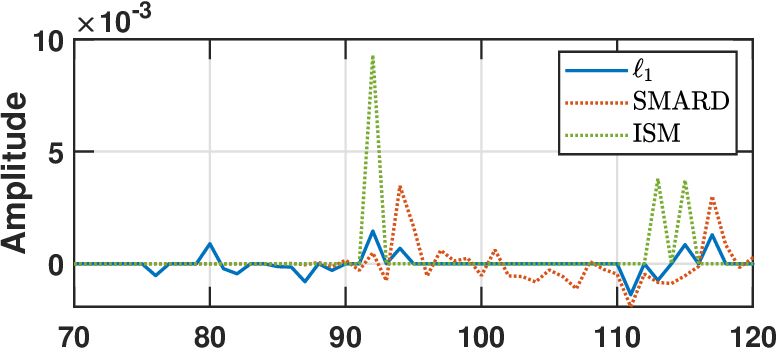}
    \caption{}\label{fig:illustration_l1}
\end{subfigure}
\caption{Illustration of how the proposed method allows energy from the simulated RIR to be transported along the support, which is not the case for $\ell_1$. }
\label{fig:illustration}\vspace{-3mm}
\end{figure}

In the noted regularization approaches, the considered RIR estimation problems consider the typical situation when only the input and output signals are observable. In contrast, we here also assume the availability of some {\em a priori} knowledge of a simulated (approximating) RIR, typically formed using a 3D model of the room geometry, based on, for example, point clouds from video data (see, e.g., \cite{Hartley2004,Li2023,Dai2020}). Prior works have suggested multiple ways to simulate such an RIR given a 3D model, including the image source method (ISM) \cite{Allen1979}, ray-tracing  methods \cite{Krokstad2015,Taylor2012}, and numerical methods such as the finite element method \cite{Ihlenburg1998}. 

However, any such simulated RIR can be expected to be subject to errors caused by imprecise knowledge of the surface reflection coefficients, as well as the actual spatial location of reflectors. Thus, this prior RIR will generally contain errors in its amplitudes as well as in its delay structure.

Although uncertainty in amplitude can be modelled using standard methods, such as Tikhonov regularization, perturbations of delays are not well described in these frameworks. As to allow for this type of uncertainty, we here propose to incorporate information contained in a prior RIR by means of an optimal transport (OT) formulation. Concerned with finding the most efficient way to transform one non-negative distribution into another \cite{Villani2007}, OT has successfully been used in signal processing applications such as spectral estimation \cite{Georgiou2009, Elvander2020b,Elvander2020a,Elvander23_icassp, haaslerE24_css_arxiv}, imaging \cite{Karlsson2017}, as well as recently for RIR tracking and interpolation \cite{Sundstrom2024ICASSP,Sundstrom2023,Geldert2023}. Herein, we use OT to model shifts in the temporal energy distribution of an RIR, allowing the model to describe both perturbations in amplitudes and delays. Furthermore, we present an efficient numerical solver using a proximal splitting approach, implementing the proposed estimator.

\section{Signal model}\vspace{-1mm}
\label{sect:signal_model}
Consider the sound field $y(t,\mathbf{r}_s, \mathbf{r}_r)$ in a room at time $t\in \mathbb{R}$ and position $\mathbf{r}_r\in \mathbb{R}^3$, generated by a  source emitting the signal $x(t)$ at position $\mathbf{r}_s$. The sound field at a position in the room may be described in terms of the RIR $h(t,\mathbf{r}_s, \mathbf{r}_r)$, such that 
\begin{equation}
    y(t,\mathbf{r}_s, \mathbf{r}_r) = h(t,\mathbf{r}_s, \mathbf{r}_r)\ast x(t),
\end{equation}
where the RIR $h(t,\mathbf{r}_s, \mathbf{r}_r)$ models the propagation of the source signal and $\ast$ denotes the convolution operator. Generally, the room RIR is determined by the surface geometry and materials, as well as properties of the propagation medium, i.e., the temperature and humidity of the air. Specifically, consider a discrete RIR, represented in terms of amplitude-delay tuples $\mathbf{h} = \{(o_k,\tau_k) \}_k$, where we here omit the notation of the source and receiver positions for notational brevity. Then, the direct sound field contributes to the tap with a delay that reflects the distance between the source and receiver, i.e., 
\begin{equation}
    \tau_{direct} = \frac{||\mathbf{r}_r-\mathbf{r}_s||_2}{c},
\end{equation}
where $c$ denotes the speed of the sound propagation.
Subsequent components of $\bh$ correspond to delays resulting from a sequence of reflections on room boundaries, as well as objects in the room. That is, a delay $\tau_k$ results from a sequence of $I$ reflections on reflectors at positions $\{\mathbf{q}_i\}_{i\in[1,\hdots,I]}$, where $\mathbf{q}_i\in \mathbb{R}^3$ according to
\begin{equation}
    \tau_{k} \!=\! \frac{1}{c}\!\left(\!\|\mathbf{r}_s-\mathbf{q}_1\|_2 \!+\! \|\mathbf{q}_{I}-\mathbf{r}_r\|_2 \!+\! \sum_{i = 1}^{I-1} \|\mathbf{q}_{i}-\mathbf{q}_{i+1}\|_2 \!\right).
    \label{eq:delay_reflection}
\end{equation}
Consequently, perturbations of any assumptions on sound speed, room geometry, and the source and receiver positions introduce a perturbation of the delay $\tau_k$, whereas deviations from the assumptions of the reflection properties instead can be expected to affect the amplitude $o_k$. While there are various methods for simulating an RIR from a room model, one may thus expect  errors in both the delay and amplitude for each reflection, as well as in the number of  reflections, when the room model is an approximation of a real room. Figure~\ref{fig:illustration} shows an illustrative example, where the early part of a measured RIR from the SMARD data set \cite{Nielsen2014} is shown alongside a simulated RIR using the ISM \cite{Allen1979}, illustrating the noted discrepancy in both delays and amplitudes. 

\section{Method}\vspace{-1mm}
\label{sect:method}

In this work, we consider the problem of estimating an RIR from measured signals, beneficially incorporating  a simulated RIR from an approximate room model. 
In the following, we consider a discrete-time setting, with a finite-length RIR $\mathbf{h}\in \mathbb{R}^{N_h}$.
Then, given an input signal $\mathbf{x}\in \mathbb{R}^N$, the signal recorded at the receiver, where we for ease of notation drop the dependency on $\mathbf{r}_s$ and $\mathbf{r}_r$, may be expressed as 
\begin{equation}
    \mathbf{y} = \mathbf{x}\ast \mathbf{h} + \mathbf{e},
    \label{eq:signal_model}
\end{equation} 
where $\mathbf{e}$ denotes an additive noise term. For well-posed settings, i.e., when the signals $\mathbf{x}$ and $\mathbf{y}$ are long, have good spectral excitation, and high SNR, the estimation problem may be posed as a least squares problem such that 
\begin{equation}
    \minwrt[\mathbf{h}] \quad \frac{1}{2} ||\mathbf{y}-\mathbf{h}\ast \mathbf{x}||_2^2.
    \label{eq:LS}
\end{equation}
However, for ill-posed settings, i.e., when, for example, the input signal is sparse in the frequency domain, the problem in \eqref{eq:LS} has to be regularized with some prior information to provide a unique solution, or to improve the conditioning of the problem, such that the problem may  be expressed as 
\begin{equation}
    \minwrt[\bh] \quad \frac{1}{2} ||\mathbf{y}-\bh\ast \mathbf{x}||_2^2+ \eta \mathcal{R}(\bh),
    \label{eq:LS_reg}
\end{equation}
where $\mathcal{R}: \mathbb{R}^{N_h}\rightarrow \mathbb{R}$ is a regularization function. 
Here, letting $\mathcal{R}$ be the (squared) $\ell_2$-norm corresponds to the standard Tikhonov regularization (see, e.g. \cite{Waterschoot2008, Ljung2020}). In the context of RIR estimation, the $\ell_1$-norm has also been used, motivated by the assumed sparse delay structure \cite{CroccoB15_eusipco}.
As an alternative, we here consider the setting in which a prior RIR, $\mathbf{h}_0$, for instance generated by a simulation, is available.
As outlined in Section~\ref{sect:signal_model}, the geometrical model from which $\mathbf{h}_{0}$ is simulated from typically contains errors with respect to the true room geometries and reflection coefficients. The naive approach for using $\mathbf{h}_{0}$ as a prior would be to minimize the difference in terms of the $\ell_p$ norms in \eqref{eq:LS_reg}, i.e., using
\begin{equation}
    \mathcal{R}_{p,\mathrm{0}}(\bh) = ||\bh-\mathbf{h}_{0}||_p^p.
\end{equation}
Although this choice is sensible for errors in the simulated amplitudes, the simulated RIR also contains errors in the delay structure. This affects the structure of the support, i.e., the location of non-zero elements, of the RIRs, which is not suitably modeled using $\ell_p$ distance measures.

Herein, we propose to use the concept of OT in order to model these types of shifts, and specifically shifts in the energy structure of the RIR. For two non-negative vectors $\bnu_1 \in \RR_+^{N_1}$, $\bnu_2 \in \RR_+^{N_2}$, the discrete Monge-Kantorovich problem of OT \cite{Peyre2019,Cuturi2013} may be stated as
\begin{equation} 
\begin{aligned}
    \minwrt[\bM \in \RR_+^{N_1 \times N_2}]& \quad \langle \bC , \bM\rangle = \mathrm{trace}\left( \bC^T \bM \right)
    \\\text{s.t. }& \quad \bM \onevec_{N_2} = \bnu_1 \;,\; \bM^T \onevec_{N_1} = \bnu_2,
\end{aligned}
\label{eq:ot_problem}
\end{equation}
where $\onevec_{N_1}$ and $\onevec_{N_2}$ are vectors of all ones of length $N_1$ and $N_2$, respectively. Here, the matrix $\bC \in \RR^{N_1 \times N_2}$ describes the cost of transporting mass between the different elements of $\bnu_1$ and $\bnu_2$. The corresponding optimal $\bM$ is the so-called transport plan describing how mass is moved between $\bnu_1$ and $\bnu_2$. In the context of RIR estimation, the minimal objective of \eqref{eq:ot_problem} has been used as a measure of distance for tracking time-varying RIRs \cite{Sundstrom2024ICASSP}, with the elements of the cost matrix being defined as $[\bC]_{k,\ell} = (\tau_k^{(1)} - \tau_\ell^{(2)})^2$, i.e., corresponding to delay discrepancies.
As may be noted, \eqref{eq:ot_problem} requires that $\bnu_1$ and $\bnu_2$ are non-negative. Furthermore, this is a linear program that, when used in the inverse problem setting considered herein, will be computationally cumbersome to solve. In order to construct a regularizing function applicable to RIRs (which can have arbitrary sign structure) as well as amenable to efficient solution, we instead propose to use $\otdist(\cdot,\mathbf{h}_{0}): \RR^{N_h} \to \RR$, defined as
\begin{equation}
\begin{aligned}
    \otdist(\mathbf{h},\mathbf{h}_0) = \minwrt[\bM \in \RR_+^{N_h \times N_h}] \quad & \langle \mathbf{C},\mathbf{M}\rangle + \epsilon D(\mathbf{M})\\
    \text{s.t. } & \mathbf{M}\mathbf{1} = \mathbf{h}_{0}^2\;,\; \mathbf{M}^T\mathbf{1} \geq \mathbf{h}^2,
\end{aligned}
\label{eq:ot_regularizer}
\end{equation}
where $D(\bM) = \sum_{k,\ell} [\bM]_{k,\ell} \log[\bM]_{k,\ell} - [\bM]_{k,\ell} + 1$ is an entropic regularization term, with $\epsilon> 0$, and where powers and inequalities are evaluated element-wise. Thus, $\otdist(\bh,\mathbf{h}_0)$ measures the effort required to rearrange the energy profile\footnote{In fact, as $\epsilon \to 0^+$, the minimal value of \eqref{eq:ot_regularizer} converges to that of the corresponding non-entropy-augmented problem \cite{Cuturi2013}.} of $\bh$ as to match that of $\mathbf{h}_0$. The relaxation of equality to inequality of the second constraint is done to make $\otdist(\cdot,\mathbf{h}_0)$ a convex function, enabling an efficient implementation. It also enables the estimated RIR to have a different total energy than the simulated RIR, which is not the case for the traditional OT problem in \eqref{eq:ot_problem}. 
Using this measure, the sought RIR, $\bh$, is estimated as
\begin{equation}
    \hat{\bh} = \argminwrt[\bh\in \RR^N_h]\quad \frac{1}{2}||\mathbf{y}- \bh \ast \mathbf{x} ||_2^2  + \eta \otdist(\bh,\mathbf{h}_0),
    \label{eq:prob}
\end{equation}
where $\eta>0$ denoted a regularization parameter determining the trade-off between data fit and the trust in the prior RIR $\mathbf{h}_0$. It may be noted that in contrast to $\ell_p$-norms, the regularizer $\otdist(\cdot,\mathbf{h}_0)$ allows for the flexibility to incorporate further prior knowledge of an expected concentration of energy, as determined by $\mathbf{h}_0$. In particular, small perturbations in the delay structure can be exploited as information by $\otdist(\cdot,\mathbf{h}_0)$.
As \eqref{eq:prob} is a convex problem, it allows for an efficient implementation. Our proposed implementation is inspired by \cite{Karlsson2017} and employs a
forward-backward splitting, separating the objective into a differentiable and a "proxable" part \cite{Beck2009}. In particular, we let the data fit term and the OT regularizer be the differentiable and proxable parts, respectively. With this, we propose solving \eqref{eq:prob} using the proximal gradient scheme\footnote{It may be noted that the convergence rate of these iterations may be improved in a straight-forward manner by means of acceleration methods \cite{Beck2009}. Furthermore, as the gradient step only involves applying convolution and its adjoint, it can be implemented using the Fast Fourier Transform \cite{Jalmby2023}.}
\begin{align*}
    \bh^{(j+1)} &= \mathrm{prox}_{\gamma\eta \otdist(\cdot,\bhgeo)}\left( \bh^{(j)} - \gamma\nabla_{\bh} \left(\frac{1}{2}||\mathbf{y}-\mathbf{X}\bh^{(j)} ||_2^2 \right) \right)
    \\&= \mathrm{prox}_{\gamma\eta \otdist(\cdot,\bh_{\mathrm{0}})}\left( \bh^{(j)} - \gamma \mathbf{X}^T (\mathbf{X}\bh^{(j)} - \mathbf{y}) \right),
\end{align*}
where $j$ denotes the iteration number, $\gamma > 0$ the stepsize, $\mathbf{X}$  the equivalent matrix representation of the convolution operator, and $\mathrm{prox}_{\gamma\eta \otdist(\cdot,\bhgeo)}$  the proximal operator for $\gamma\eta \otdist(\cdot,\bhgeo)$. We here set the stepsize as $\gamma = 1/L$, where $L = \| \mathbf{X} \|^2$ is the Lipschitz constant for the data fit term, and $\|\cdot \|$ denotes the operator norm.
Furthermore, the proximal operator is given by the following proposition.
\begin{proposition}
    For any $\theta > 0$, the proximal operator for $\theta \otdist(\cdot,\bhgeo):\RR^{N_h} \to \RR$  is unique and given by
\begin{align*}
    \mathrm{prox}_{\theta \otdist(\cdot,\bhgeo)}(\bu) = \bu \oslash \left( 2\bmu + \onevec \right),
\end{align*}
    where $\onevec \in \RR^{N_h}$ is a vector of all 1's, $\oslash$ denotes elementwise division,  $\otimes$ is the Kronecker product, and $\bmu \in \RR_+^{N_h}$ solves
\begin{equation} \label{eq:prox_dual_problem}
    \begin{aligned}
        \minwrt[\bmu \in \RR_+^{N_h} \:,\: \blambda \in \RR^{N_h}]\!\! \theta \epsilon \left\langle \bK , \bv \otimes \bw\right\rangle 
        \!-\! \langle \bhgeo^2, \blambda \rangle \!-\! \langle \bu^2 , \bmu\oslash (\onevec + 2\bmu) \rangle,
    \end{aligned}
\end{equation}    
     where $\bw = \exp\left( \frac{1}{\theta \epsilon} \bmu \right)$, $\bv = \exp\left( \frac{1}{\theta \epsilon} \blambda \right)$, $\bK = \exp\left( -\frac{1}{\epsilon} \bC \right)$.
     Here, all exponentiation and powers are element-wise.
\end{proposition}
\begin{proof}
    See appendix.
\end{proof}
The proximal operator does not have an analytical solution and has to be computed using an iterative scheme solving \eqref{eq:prox_dual_problem}. We propose to address this using block-coordinate descent, with the blocks corresponding to $\bmu$ and $\blambda$. In particular, in iteration $k$, the updates are given by (see the appendix)
\begin{align*}
    \blambda^{(k)} &= \theta \epsilon\left( \log \bhgeo^2 -  \log\left( \bK \bw^{(k-1)} \right)\right),
    \\\bmu^{(k)} &= 2\theta \epsilon\left( \omega\left( \bxi^{(k)} \right) - \frac{1}{4\theta\epsilon} \onevec\right)_+,
\end{align*}
where $\omega(\cdot)$ denotes the (element-wise) Wright omega function \cite{Corless2002}, $(\cdot)_+$ element-wise truncation at zero, and where 
\begin{align*}
    \bxi^{(k)} = \left(\frac{1}{4\theta\epsilon} -\log(4\theta\epsilon)\right)\onevec + \frac{1}{2} \log \bu^2 - \frac{1}{2}\log \bK^T \bv^{(k)}.
\end{align*}
As the dual problem \eqref{eq:prox_dual_problem} satisfies the assumptions of  \cite[Theorem 2.1]{luo1992convergence}, the iterates converge linearly to the solution of \eqref{eq:prox_dual_problem}.
It may be noted that the scheme can be warm-started by using the previous optimal pair $(\bmu,\blambda)$ as the initial point of the iterations.
Empirically, we observe fast convergence of the proposed scheme. 
\section{Numerical experiments}\vspace{-1mm}
We proceed to evaluate the proposed method on both simulated and measured RIRs from the SMARD data set \cite{Nielsen2014}. The proposed method, using $\epsilon = 0.1$, is compared to the state-of-the-art methods described in Section~\ref{sect:method}, i.e., Tikhonov and Lasso regularization, and using the $\ell_2$ and $\ell_1$ distance to the simulated RIR as regularization. The regularization parameter $\eta$ is for all methods set by cross-validation with 30 logarithmically spaced values in the range $10^{-6}$ to $10^{6}$. To evaluate the performance of the estimated RIRs, we use the NMSE, defined as
\begin{equation}
    \text{NMSE} = \sum_{k=1}^K \frac{||\hat{\mathbf{h}}_k \ast \mathbf{z} - \mathbf{h}_k \ast \mathbf{z}||_2^2}{||\mathbf{h}_k \ast \mathbf{z}||_2^2},
    \label{eq:NMSE}
\end{equation}
where $\hat{\mathbf{h}}_k$ denotes the estimated RIR, $\mathbf{h}_k$ the true RIR, $K$ the number of realizations of the numerical experiment, and $\mathbf{z}$ a low-pass filter with cut-off frequency $
3000$~Hz introduced to avoid small deviation in the delays to cause magnitude contributions to the NMSE (see also for example \cite{Geldert2023}). 

\begin{figure}
    \centering
    \includegraphics[width = \linewidth]{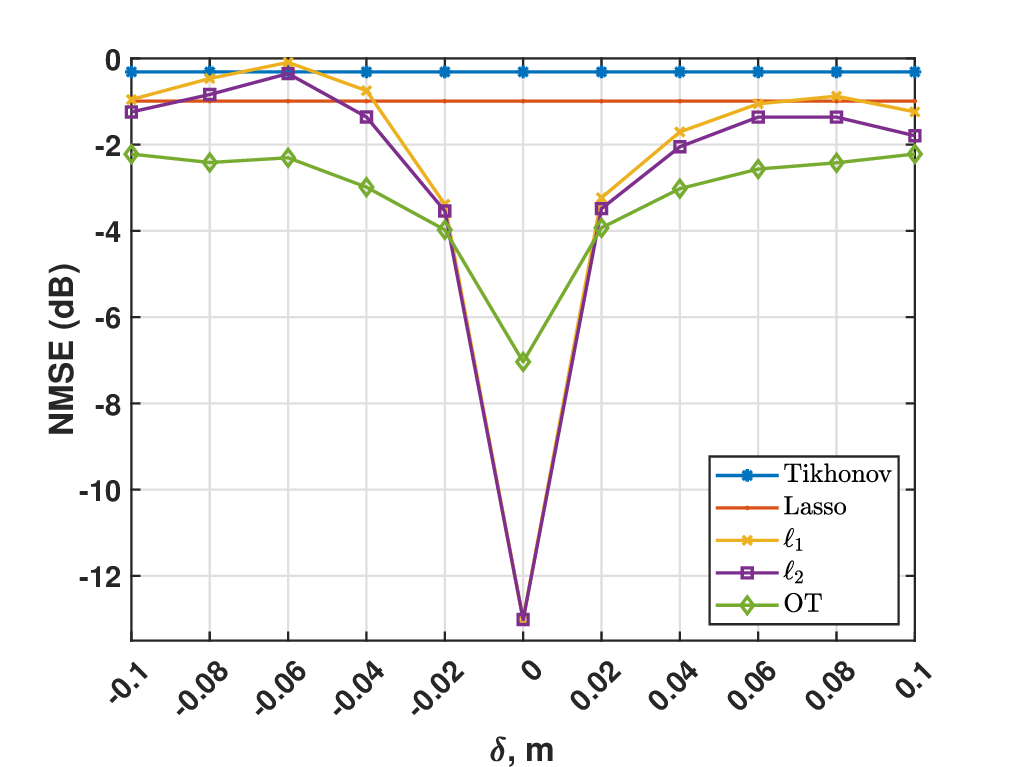}
    \caption{Robustness 
    to errors in the room dimensions in the simulated RIR $\mathbf{h}_{0}$ using signals generated from the ISM.}
    \label{fig:robustness_room_dims}\vspace{-1mm}
\end{figure}

We begin by evaluating the robustness to model errors in the simulated RIR, $\mathbf{h}_{0}$, for simulated RIRs. The observed signals are generated by a simulated RIR of length $600$ at the sampling frequency $8$~kHz using the ISM as implemented in \cite{Govern2009}, with the reflection coefficient $0.5$, room dimensions $7\times5\times3$~m, the temperature $19.6$~$^\circ$C, and with a source positioned at $[5,4,1]$. The performance is averaged over $K=10$ microphone positions randomly generated in a cube of side length $1$~m centered at $[2,2,1.5]$. 
As a source signal a $12.5$~ms long section of a real speech recording is used, where a new section of the recording is used for each new microphone position. Furthermore, white Gaussian noise is added to the microphone signal to achieve a signal-to-noise ratio (SNR) of $5$, with the SNR defined as 
$\mathrm{SNR} = 10\log_{10}\left( \sigma^2_{\mathrm{signal}} / \sigma^2_{\mathrm{noise}} \right),$
where $\sigma^2_{\mathrm{signal}}$ and $\sigma^2_{\mathrm{noise}}$ denote the power of the signal $\mathbf{x} * \mathbf{h}$ and the noise $\mathbf{e}$, respectively.
The least squares problem in \eqref{eq:LS} is thus ill-conditioned both in terms of the short signal length, poor spectral excitation of the speech signal, and its low SNR. 

In practice, for the applications outlined in Section~\ref{sect:introduction}, errors in both the temperature and the room geometry are inevitable. In Figure~\ref{fig:robustness_room_dims}, the robustness with respect to errors in the room dimensions are illustrated, where $\mathbf{h}_{0}$ is identical to $\mathbf{h}$ except for an additive perturbation $\delta$ in each room dimension for $11$ equally spaced values of $\delta$ in the range $-0.1$ to $0.1$~m. Similarly, Figure~\ref{fig:robustness_temperature} illustrates the robustness to errors in the temperature of $\mathbf{h}_{0}$ for $11$ equally spaced values of the temperature in the range $-14.6$ to $24.6$~$^\circ$C. For an ideal $\mathbf{h}_0$, i.e., when $\delta$ is $0$~m in Figure~\ref{fig:robustness_room_dims}, and the temperature is $19.6$~$^\circ$C in Figure~\ref{fig:robustness_temperature}, the $\ell_1$ and $\ell_2$ methods have the lowest NMSE. However, in both cases it is clear that the proposed method is more robust to errors in $\mathbf{h}_{0}$ as compared to the $\ell_1$ and $\ell_2$ regularizers. 
While the simulated experiments illustrated in Figures~\ref{fig:robustness_room_dims} and \ref{fig:robustness_temperature} isolate one type of error at the time, in a real setting one may expect different kinds of errors to simultaneously influence the estimation.
\begin{figure}
    \centering
    \includegraphics[width = \linewidth]{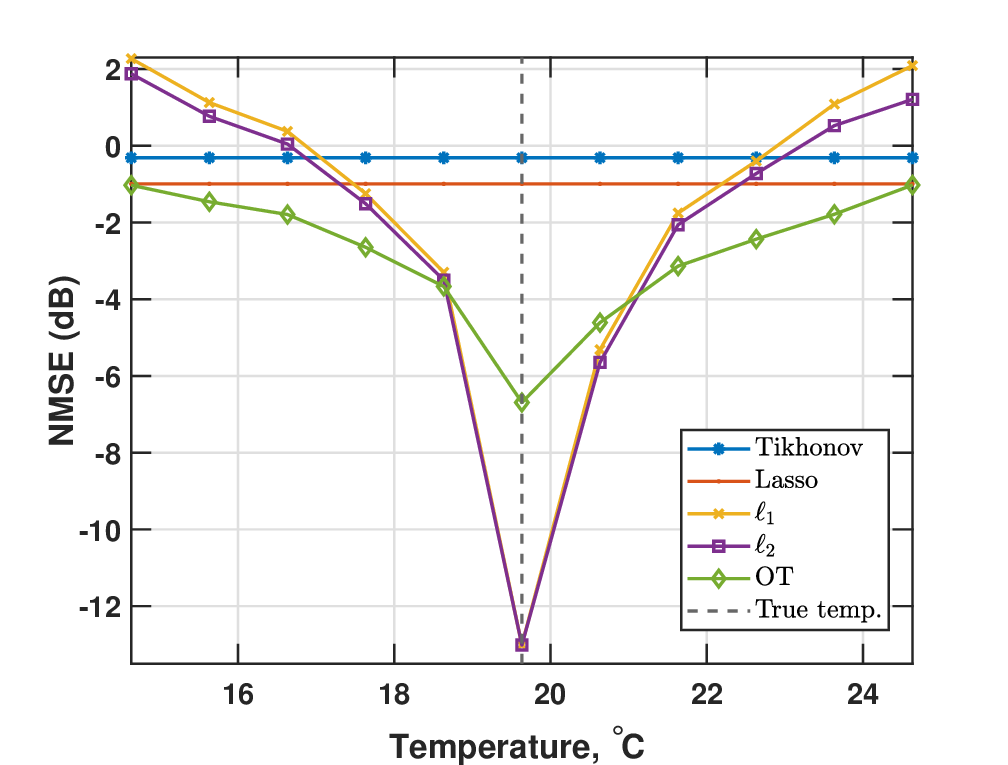}
    \caption{Robustness 
    to errors in the temperature in the simulated RIR $\mathbf{h}_{0}$ using signals generated from the ISM.}
    \label{fig:robustness_temperature}\vspace{-1mm}
\end{figure}
Finally, we estimate RIRs from the SMARD data set \cite{Nielsen2014} to validate the performance for a realistic scenario. We use RIRs downsampled to $8$~kHz from the subset $1002$, including RIRs to microphones in three linear arrays. As the simulated RIR, we use what could be considered as the most simplistic simulation, i.e., the ISM in \cite{Govern2009}, using the room dimension, temperature source position, and microphone position documented in the data set. While the reflection coefficient on the other hand is both unknown and in practice varying for every surface, we set it somewhat arbitrary to $0.3$ to reflect the short reverberation time of $0.15$~s documented in the data set. As illustrated in Figure~\ref{fig:illustration}, the simulated RIR is a clear simplification  of the real RIR.
The input signal is similar to the one used above and is observed with a SNR of $15$~dB, with the performance being measured as the average over $10$ realizations of microphone positions and sections of the speech signal. The robustness with respect to the choice of temperature in the simulation model is illustrated Figure~\ref{fig:temperature_SMARD}, where the proposed method has the lowest NMSE, even for large errors in the temperature. We also confirm that the results of both the $\ell_1$ and $\ell_2$ methods are not meaningful, indicating that also other types of errors are present.
\begin{figure}[t]
    \centering
\includegraphics[width=\linewidth]{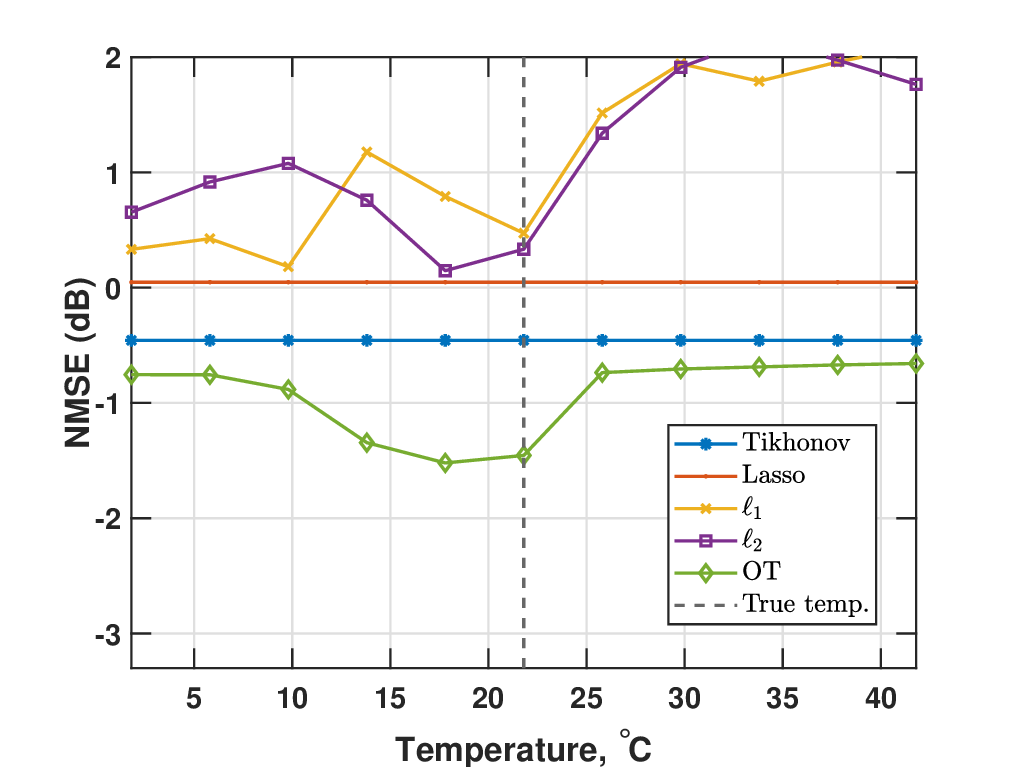}
    \caption{Robustness with respect errors in the temperature in the simulated RIR $\mathbf{h}_{0}$ using signals generated from the SMARD data set.}
    \label{fig:temperature_SMARD}\vspace{-1mm}
\end{figure}
\section{Conclusion}\vspace{-1mm}
In this work, we consider the problem of estimating an RIR using ambient signals, such as speech and music, when approximate knowledge of the delay structure of the RIR is available. We employ an optimal transport regularization technique to allow for differences in the delay structure 
and propose an efficient numerical solver for the resulting estimator.
Using simulated and measured data, it is shown that the proposed method is able to beneficially incorporate information from even a simple simulation model, yielding robustness to errors in both the room dimensions and temperature. 

\appendix
\section{Proof of Proposition 1}
\begin{proof}
The proximal operator for $\theta \otdist(\cdot,\bhgeo)$ is defined as
    \begin{equation*}
        \text{prox}_{\theta \otdist(\cdot, \bhgeo)}(\mathbf{u}) = \argmin_{\mathbf{h}\in \mathbb{R}^{N_h}}  \theta \otdist(\bh, \bhgeo) + \frac{1}{2} ||\mathbf{u}-\mathbf{h} ||_2^2,
    \end{equation*}
    where
    \begin{equation*}
    \begin{aligned}
        \otdist(\mathbf{h},\mathbf{h}_0) = \minwrt[\bM \in \RR_+^{N_h \times N_h}] &\quad \langle \mathbf{C},\mathbf{M}\rangle + \epsilon D(\mathbf{M})\\
        \text{s.t. } &\quad \mathbf{M}\mathbf{1} = \mathbf{h}_{0}^2\;,\; \mathbf{M}^T\mathbf{1} \geq \mathbf{h}^2.
    \end{aligned}
    \end{equation*}
That is, the proximal operator solves
\begin{equation} \label{eq:full_problem}
    \begin{aligned}
        \minwrt[\mathbf{h} \in \RR^{N_h}, \mathbf{M} \in \RR^{N_h\times N_h}]& \quad\theta \langle \mathbf{C},\mathbf{M}\rangle + \theta\epsilon D(\mathbf{M})) + \frac{1}{2} ||\mathbf{u}-\mathbf{h} ||_2^2
        \\\text{s.t. } &\quad \mathbf{M}\mathbf{1} = \mathbf{h}_{0}^2\;,\; \mathbf{M}^T\mathbf{1} \geq \mathbf{h}^2.
    \end{aligned}
\end{equation}
The Lagrangian of \eqref{eq:full_problem} is given by
    \begin{align*}
        \mathcal{L}(\mathbf{M},\mathbf{h},\blambda,\bmu) &= \theta\langle \mathbf{C},\mathbf{M}\rangle + \theta \epsilon D(\mathbf{M}) + \left\langle \blambda, \bhgeo^2 - \bM \mathbf{1} \right\rangle \\
        &+ \left\langle \bmu, \bh^2 - \bM^T \mathbf{1} \right\rangle \notag + \frac{1}{2} \|\bu - \bh\|_2^2, \notag
    \end{align*}
    where $\bmu\in \mathbb{R}^{N_h}_+$ and $\blambda\in \mathbb{R}^{N_h}$ are the dual variables. It may be readily verified that the Lagrangian is strongly convex in $\bh$ and $\bM$ with the unique minimizer
    \begin{align*}
        \bh &= \bu \oslash (2\bmu + \mathbf{1}), 
        \\\bM &= \text{diag}(\bv) \bK \text{diag}(\bw) = \bK \odot (\bv \otimes \bw), 
    \end{align*}
where $\odot$ is the Hadamard product. Plugging this into the Lagrangian yields the dual problem
    \begin{equation} \label{eq:dual_problem}
    \begin{aligned}
        &\maxwrt[\blambda \in \mathbb{R}^{N_h}, \bmu \in \mathbb{R}_+^{N_h}] \left\langle \blambda,\bhgeo^2 \right\rangle + \frac{1}{2}
        \left\| \bu \oslash (2\bmu + \boldsymbol{1}) - \bu \right\|^2_2\\
        &+ \left\langle \bmu,\bu^2 \oslash (2\bmu + 1)^2 \right\rangle - \epsilon \theta \bw^T \bK \bv + \epsilon \theta N_h^2.
    \end{aligned}
    \end{equation}
    Simplifying and omitting constant terms, we arrive at the minimization problem
    \begin{equation*}
    \begin{aligned}
        \minwrt[\blambda \in \mathbb{R}^{N_h}, \bmu \in \mathbb{R}_+^{N_h}] \theta \epsilon \left\langle \bK , \bv \otimes \bw\right\rangle - \langle \bhgeo^2, \blambda \rangle \!-\! \langle \bu^2 , \bmu\oslash (\onevec + 2\bmu) \rangle,
    \end{aligned}
    \end{equation*}
which has the same solution as \eqref{eq:dual_problem}.
\end{proof}

\section{Block coordinate descent iterations}
\begin{proof}
Keeping $\bmu$ fixed, minimizing \eqref{eq:dual_problem} with respect to $\blambda$ is equivalent to solving
\begin{align*}
    \minwrt[\blambda \in \RR^{N_h}] \quad \theta \epsilon \langle \bK \bw, \bv \rangle - \langle \bhgeo^2,\blambda \rangle,
\end{align*}
where $\bw = \exp\left( \frac{1}{\theta \epsilon} \blambda \right)$. Thus, the optimal $\blambda$ is found by solving the zero gradient equations,
\begin{align*}
    \bv \odot \bK \bw - \bhgeo^2 = 0
\end{align*}
yielding
\begin{align*}
    \blambda = \theta \epsilon \left( \log \bhgeo^2 - \log \bK \bw \right).
\end{align*}
Keeping $\blambda$ fixed, minimzing with respect to $\bmu$ is equivalent to solving
\begin{align*}
    \minwrt[\bmu \in \RR_+^{N_h}] \quad \theta \epsilon \langle \bK^T \bv, \bw \rangle - \langle \bu^2 , \bmu\oslash (\onevec + 2\bmu) \rangle,
\end{align*}
with $\bw = \exp\left( \frac{1}{\theta \epsilon} \bmu \right)$. As may be noted, this problem decouples in the individual component of $\bmu$ and can this be solved for each element separately. Let $\mu$ be a components of $\bmu$, and let $q$ and $u$ be corresponding elements of $\mathbf{q} = \bK^T\bv$ and $\bu$. This yields the problem
\begin{align*}
    \minwrt[\mu \geq 0] \quad f(\mu) =  \theta \epsilon \exp\left( \frac{1}{\theta \epsilon} \mu\right) q - u^2  \frac{\mu} {1 + 2\mu}.
\end{align*}
As this problem is convex for $\mu \geq 0$, it follows directly that the minimizer $\mu^\star$ is given as
\begin{align*}
    \mu^\star = \max\left( 0, \mu_0\right) = \left( \mu_0\right)_+,
\end{align*}
where $\mu_0$ is a root of the derivative of $f$. This derivative is given by
\begin{align*}
    f'(\mu) = \exp\left( \frac{1}{\theta \epsilon} \mu \right) q - u^2 \frac{1}{(1+2\mu)^2}.
\end{align*}
Setting this to zero yields
\begin{align*}
    &\quad\exp\left( \frac{1}{\theta \epsilon} \mu_0 \right) q=  u^2 \frac{1}{(1+2\mu_0)^2} \\\iff&\quad
\frac{1}{2\theta \epsilon}\mu_0  +\log(1+2\mu_0)= \frac{1}{2}(\log u^2 - \log q),
\end{align*}
under the assumption $\mu_0 > -1/2$. Adding $1/(4\theta) - \log(4\theta)$ to both sides of this equation yields
\begin{align*}
    \frac{1}{4\theta \epsilon} + \frac{1}{2\theta \epsilon}\mu_0 + \log\left( \frac{1}{4\theta \epsilon} + \frac{1}{2\theta \epsilon}\mu_0\right) = \xi
\end{align*}
where
\begin{align*}
   \xi =  \frac{1}{4\theta \epsilon} - \log(4\theta \epsilon) + \frac{1}{2}(\log u^2 - \log q),
\end{align*}
and thus
\begin{align*}
    \frac{1}{4\theta \epsilon} + \frac{1}{2\theta \epsilon}\mu_0 + \log\left( \frac{1}{4\theta \epsilon} + \frac{1}{2\theta\epsilon}\mu_0\right) = \omega(\xi)
\end{align*}
where $\omega(\cdot)$ is the Wright omega-function, i.e., the function $\omega: \RR \to \RR_+$ mapping $x$ to $\omega(x)$ such that $\omega(x) + \log \omega(x) = x$. From this, we can conclude that
\begin{align*}
    \mu_0 = 2\theta \epsilon\left( \omega(\xi) - \frac{1}{4\theta\epsilon}\right)
\end{align*}
is the (unique) root of the derivative. Thus,
\begin{align*}
    \mu^\star = 2\theta\epsilon\left( \omega(\xi) - \frac{1}{4\theta\epsilon}\right)_+,
\end{align*}
which when applied element-wise yields the optimal $\bmu$ as
 \begin{align}
        \bmu &= 2\theta \epsilon\left( \omega\left( \bxi \right) - \frac{1}{4\theta\epsilon} \onevec\right)_+.
    \end{align}
As the objective function satisfies the assumptions of \cite[Theorem 2.1]{luo1992convergence}, such as, e.g., strong convexity, the iterates constructed by alternatingly minimizing with respect to $\blambda$ and $\bmu$ converges linearly to the solution of \eqref{eq:prox_dual_problem}.
\end{proof}

\bibliographystyle{IEEEbib}
\bibliography{eusipco2024_arxiv.bbl}

\end{document}